\documentclass[a4paper,UKenglish]{lipics-v2018}

\usepackage{microtype}
\usepackage{amsfonts,amsmath,amsthm}
\usepackage{xspace}
\usepackage{soul}
\usepackage{url}

\usepackage[noend]{algpseudocode}
\usepackage{algorithm}
\usepackage{comment}

\usepackage{multirow,booktabs}
\usepackage{mathrsfs}
\allowdisplaybreaks
\usepackage{xcolor}
\usepackage{float}
\usepackage[]{graphicx}
\usepackage[]{hyperref}
\usepackage{caption}
\newtheorem{claim}{Claim}
\newtheorem{result}{Result}

\newcommand{\mcL}{\mathcal{L}}

\newcommand{\mcP}{\mathcal{P}}
\newcommand{\mcT}{\mathcal{T}}
\newcommand{\mcI}{\mathcal{I}}
\newcommand{\real}{\mathbb{R}}

\title{Computational Aspects of Equilibria in Discrete Preference Games\footnote{Part of the work was done while the first author was an intern at IBM Research, India. Work funded in part by an OSCP joint project between IBM Research, India and TIFR, and a Ramanujan award and an Early Career Research award of the second author.}}

\author{Phani Raj Lolakapuri}{Tata Institute of Fundamental Research, Mumbai, India}{p.lolakapuri@tifr.res.in}{}{}

\author{Umang Bhaskar}{Tata Institute of Fundamental Research, Mumbai, India}{umang@tifr.res.in}{}{}

\author{Ramasuri Narayanam}{IBM, Bangalore, India}{ramasurn@in.ibm.com}{}{}
\author{Gyana R Parija}{IBM, Bangalore, India}{gyana.parija@in.ibm.com}{}{}
\author{Pankaj S Dayama}{IBM, Bangalore, India}{pankajdayama@in.ibm.com}{}{}

\authorrunning{P.\,R. Lolakapuri, U. Bhaskar, R. Narayanam, G.\,R. Parija and P.\,S. Dayama}

\Copyright{Umang Bhaskar and Phani Raj Lolakapuri}

\subjclass{Theory of computation $ \rightarrow $  Network games}

\category{}

\relatedversion{}

\supplement{}

\EventEditors{Yossi Azar, Hannah Bast, and Grzegorz Herman}
\EventNoEds{3}
\EventLongTitle{26th Annual European Symposium on Algorithms (ESA 2018)}
\EventShortTitle{ESA 2018}
\EventAcronym{ESA}
\EventYear{2018}
\EventDate{August 20--22, 2018}
\EventLocation{Helsinki, Finland}
\EventLogo{}
\SeriesVolume{112}
\ArticleNo{58}
\nolinenumbers 
\hideLIPIcs  

\begin{document}

\maketitle

\begin{abstract}
We study the complexity of equilibrium computation in discrete preference games. These games were introduced by Chierichetti, Kleinberg, and Oren (EC '13, JCSS '18) to model decision-making by agents in a social network that choose a strategy from a finite, discrete set, balancing between their intrinsic preferences for the strategies and their desire to choose a strategy that is `similar' to their neighbours. There are thus two components: a social network with the agents as vertices, and a metric space of strategies. These games are potential games, and hence pure Nash equilibria exist. Since their introduction, a number of papers have studied various aspects of this model, including the social cost at equilibria, and arrival at a consensus.

We show that in general, equilibrium computation in discrete preference games is PLS-complete, even in the simple case where each agent has a constant number of neighbours. If the edges in the social network are weighted, then the problem is PLS-complete even if each agent has a constant number of neighbours, the metric space has constant size, and every pair of strategies is at distance 1 or 2. Further, if the social network is directed, modelling asymmetric influence, an equilibrium may not even exist. On the positive side, we show that if the metric space is a tree metric, or is the product of path metrics, then the equilibrium can be computed in polynomial time. 
\end{abstract}

\section{Introduction}
\label{introduction}


Networks are a growing presence in our lives, and affect our behaviour in complex ways. A large amount of literature attempts to understand various facets of these networks. The literature is diverse, due to the large variety of networks and their myriad effects on our daily lives. Prominent among these is the work on opinion formation in social networks~\cite{BalaG98,golub:2010}; algorithms to target agents in a network to promote adoption of a product~\cite{DomingosR01,KempeKT15}; and models that accurately capture the special structure of social networks~\cite{BarabasiA99,WattsS98}.

We study a model of opinion formation in social networks. In a basic but commonly studied model of opinion formation, each agent in the network holds a real-valued opinion, such as her political leaning, and is influenced by her neighbours in the social network. Under the influence of her neighbours, in each time step she updates her opinion to the weighted average of her opinion and that of her neighbours. In a game-theoretic setting, this is a \emph{coordination game}, where players try to coordinate their opinion with their neighbours. Probabilistic models of updation, where the opinions are from the discrete set $\{0,1\}$ are also studied~\cite{CliffordS73}. Much of the work in opinion formation focuses on conditions for consensus, when all agents eventually hold the same opinion~(e.g., \cite{AcemogluDLO11,BalaG98}). Clearly, however, consensus is not always attained in social networks, and the basic model has been extended in different ways to capture this lack of consensus (e.g.,~\cite{yildiz:2013,AptSW16,chierichetti:jcss18}).

Further, most work focuses on the case where the opinion of an agent is either binary (in the set $\{0,1\}$), or in the interval $[0,1]$. These are clearly important, since opinions in many cases (e.g., political leanings, or product adoption) are captured by these sets. However, often more complex sets are required. As an example, a person's political leaning is often a composite of her inclinations on various topics, such as economic inequality, foreign policy, and the tax regime. A more realistic model would consider a person's opinion as a composite of these individual opinions, and update accordingly. As a second example, a person's opinion could be a physical location, such as a choice of which neighbourhood to live in. In this case, the set of strategies would be more complex, and the update process would select the geometric median of the neighbour locations. Another example would be 
in understanding technology adoption from among different platforms such as Android, iOS, Blackberry, etc. The set of strategies are now discrete points, with distances corresponding to the cost of switching from one technology platform to another. 

We study a particular model for opinion formation called a \emph{discrete preference game} \cite{dpg:13,chierichetti:jcss18}.\footnote{A similar model was concurrently studied by Ferraioli et al. \cite{FerraioliGV16}, however with binary strategies. Both these papers give a natural polynomial time algorithm for equilibrium computation with binary strategies.} In this model, an agent can hold one of a finite set of strategies (opinions), and a \emph{distance function} gives the distance between any pair of strategies. A natural restriction on the distance function is that it be a metric, and hence the strategies are points in a metric space. In addition, each agent has an intrinsic \emph{preferred strategy} which is fixed. The cost of each agent for a strategy is the sum of weighted distances to her neighbours and to her preferred strategy. The presence of preferred strategies leads to the absence of consensus as an equilibrium in general \cite{krackhardt:09}. Further, the representation of strategies as points in a metric space allows modelling of many complex situations, beyond the simple settings studied earlier.

Since their introduction, numerous papers have studied various properties of these games, including bounding the ratio of the total cost of equilibria to the minimum total cost (called the Price of Anarchy or Stability), as well as generalisations~\cite{AulettaCFGP16,chierichetti:jcss18}. In a natural updation process, each player in her turn chooses a strategy that minimizes her cost. While it is known that this updation process leads to an equilibrium, the number of turns required may be exponential in the size of the game. 

In this work, we study computational aspects of equilibria in discrete preference games. Equilibrium computation is a fundamental problem in computational game theory, and the lack of efficient algorithms for this is often viewed as a stumbling block to the notion of equilibria as a prediction of player behaviour (e.g.,~\cite{DaskalakisGP09}). Algorithms for computing equilibria are also useful, e.g., in simulations to study properties of equilibria, or to obtain approximations to the global optimum for the underlying distance-minimization problem (e.g.,~\cite{BoykovVZ01}).

Coordination games on graphs are another model closely related to discrete preference games~\cite{AptKRSS17,AptSW16}. In these games, agents attempt to coordinate with their neighbours, however the set of strategies available to each player is restricted. The distance between any pair of strategies is 1, and hence these are similar to discrete preference games with the discrete metric. 



\subsection*{Our Contribution}

We present our results informally here. Formal definitions and results are given in later sections. 

We first show that equilibrium computation in discrete preference games is hard, even if we restrict the number of neighbours that each agent has in the social network.

\begin{result}
It is PLS-hard to find an equilibrium in discrete preference games, even when each player has constant degree in the social network.
\end{result}

If we allow the edges in the social network to be weighted, modelling varying degrees of influence by the neighbours, then equilibrium computation is hard even with multiple restrictions on the metric space.

\begin{result}
In weighted discrete preference games, it is PLS-hard to compute an equilibrium, even when each player has constant degree in the social network, the number of strategies is constant, and the distance between any pair of strategies is one or two.
\end{result}

Our results are interesting because these are examples where equilibrium computation is hard in a purely coordination game. In previous games where hardness was shown for equilibrium computation, there were incentives for anti-coordination, i.e., players had an incentive to choose strategies different from their neighbours (e.g., local max-cut games~\cite{SchafferY91}, congestion games~\cite{FabrikantPT04}, and even coordination-only polymatrix games ~\cite{CaiD11}).

Lastly, we show that if we allow the edges in the social network to be directed, then an equilibrium may not even exist (and hence, the update process described may cycle).

\begin{result}
In a discrete preference game with directed edges, an equilibrium may not exist.
\end{result}

We note that directed edges in social networks are clearly more general, and allow the model to capture asymmetric influences. E.g., Facebook offers one the ability to `follow' another person, which is an asymmetric method of influence. Both undirected and directed social networks are commonly studied (e.g.,~\cite{AptKRSS17,AptSW16,bindel-kleinberg:15,yildiz:2013}).

In our example to show nonexistence of equilibria, the social network consists of a single strongly connected component. In coordination games on graphs, it is known that if the social network has a single strongly connected component then an equilibrium always exists~\cite{AptSW16}. Our work thus shows this does not hold if we allow more complicated metric spaces. 

We show, however, that in two particular cases, an equilibrium can be computed in polynomial time.

\begin{result}
If the metric space is a tree metric, or is the Cartesian product of path metrics, an equilibrium can be computed in polynomial time.
\end{result}

The case of tree metrics was earlier studied, and bounds shown on the Price of Stability~\cite{chierichetti:jcss18}. The authors also motivate tree metrics by an example of students choosing a major in college, when different subjects follow a hierarchy for proximity. The product metric space roughly corresponds to the case when the metric space is a regular grid. A natural scenario that is modelled by the product metric is the case presented in the introduction, where an agent's strategy is a composite of a number of individual opinions, and the distance between two strategies is the sum of distances for each individual opinion.

Our algorithms for these cases are simple, however, they obtain equilibria in substantial generalizations of discrete preference games as well, when the social network is a weighted directed graph, and instead of having a single preferred strategy, agents have multiple preferred strategies with different weights for each. Thus, this result also shows the existence of equilibria in directed discrete preference games, with these metric spaces.

\section{Preliminaries}

In the basic model, a discrete preference game consists of an undirected, unweighted \emph{neighbourhood graph} $G=(V,E)$ representing the social network of $n$ players, and a metric space $\mcL=(L,d)$~\cite{chierichetti:jcss18}. Here, $L$ is the set of strategies, and $d$ is a distance metric ---  $d$ is a function on pairs of strategies that satisfies: (i) $d(x,y) = 0$ iff $x=y$, and is positive otherwise, (ii) $d(x,y) = d(y,x)$, and (iii) $d(x,y) \le d(x,z) + d(y,z)$. Each player $i \in V$ has a \emph{preferred strategy} $s_i \in L$. Since the strategies exist in a metric space, we will also refer to the strategies as points in the metric space. We will use $z_i$ for the strategy of the $i$th player, $z = (z_1, \ldots, z_n)$ for a strategy profile, and $z_{-i}$ for the strategies of all players except $i$. 

Given a parameter $\alpha \in [0,1)$ and a strategy profile $z$, the cost for player $i$ is:

\[
c_i(z) = \alpha d(s_i,z_i) + (1-\alpha) \sum_{j \in N_i} d(z_i, z_j) \, ,
\]

\noindent where $N_i$ is the set of neighbours of $i$, not including $i$ herself. Thus the cost of a strategy $z_i$ for player $i$ is $\alpha$ times the distance from her preferred strategy, plus $(1-\alpha)$ times the total distance from her neighbours. Each player tries to minimise her cost, and hence tries to choose a strategy that is the weighted median of her preferred strategy and the strategies of her neighbours. 


We also study two natural generalisations of the basic model of discrete preference games. In the first generalisation, we allow weights on the edges of the neighbourhood graphs. This models the realistic scenario when different neighbours of a player have different levels of influence on her actions. In this case, for player $i$, the strategy profile $z$ has cost:

\[
c_i(z) = w_i d(s_i,z_i) + \sum_{j \in N_i} w_{ij} d(z_i, z_j) \, ,
\]

\noindent where $w_i$ is the weight player $i$ places on her preferred strategy, and $w_{ij}$ is the weight on the undirected edge $\{i,j\}\in E$.

In the second generalisation, we allow edges to be directed as well as weighted. This naturally models the case when influences are asymmetric: e.g., Facebook, in addition to the option of adding a person as a friend, offers one the ability to `follow' another person, which is an asymmetric method of influence. In this case, the expression for the cost of player $i$ for strategy profile $z$ remains unchanged, though the neighbours of player $i$ are those players that have edges from $i$ in the neighbourhood graph.

An equilibrium is a strategy profile where no player can deviate to a different strategy and reduce her cost. We are interested in algorithms for computing equilibria in discrete preference games. In the weighed setting, these games are \emph{exact potential} games. That is, for every weighted discrete preference game, there is a potential function $\Phi$ of the strategy profile which has the property that if player $i$ deviates from a strategy profile, then the change in player $i$'s cost is exactly the change in the potential function as well. It can be verified that the potential function for the weighted setting is:

\begin{equation}
\Phi(z) = \sum_{i \in V} w_i d(s_i,z_i) + \sum_{\{i,j\} \in E} w_{ij} d(z_i,z_j) \,.
\label{eqn:potential}
\end{equation}

A finite potential game always has an equilibrium, since at the minimum of the potential function, no player has a deviating strategy that reduces her cost. Thus, undirected weighted discrete preference games always possess an equilibrium. Further, \emph{best response dynamics} --- where in each step, a player chooses her minimum cost strategy in response to other players, and deviates to it --- converges to an equilibrium, since in each step the potential function decreases.

However, best-response dynamics may, in general, take exponential time to converge to an equilibrium. We are interested in efficient algorithms for equilibrium computation, that for some polynomial $p(\cdot)$ run in time $O(p(|\mcI|))$ where $|\mcI|$ is the size of input, and returns an equilibrium. This is the subject of Section~\ref{sec:algorithms}. 

In Section~\ref{sec:hardness} we show that in general, the problem of equilibrium computation is hard, by showing that even in many simple cases, equilibrium computation is PLS-complete. The class PLS, for Polynomial Local Search, was introduced to study the complexity of finding a local minimum for problems where local search can be carried out in polynomial time~\cite{JohnsonPY88}. Discrete preference games fall in this class, since finding the equilibrium is equivalent to finding a local minimum for the potential function $\Phi$. The locality of a strategy profile $z$ is the set of all profiles where a single player deviates. By finding the cost of each deviation, for each player, we can obtain a solution with lower value for the potential in polynomial time, if it exists. 

A problem is PLS-complete if it is in PLS and is PLS-hard. PLS-hardness of a problem means that all problems in the class PLS can be polynomially reduced to this problem. Many problems are by now known to be PLS-complete, including local max-cut, max-2SAT, and equilibrium computation in congestion games~\cite{FabrikantPT04,SchafferY91}.


\section{Hardness of Computing Equilibria}
\label{sec:hardness}


We start with two simple cases when an equilibrium can be computed in polynomial time. Firstly, if the parameter $\alpha \le 1/2$, then in any instance where the neighbourhood graph is connected, the following is an equilibrium: all players choose the same strategy $A \in \mcL$. If the neighbourhood graph is disconnected, then each isolated player chooses its preferred strategy, while all players in a connected component choose the same strategy. Secondly, in weighted preference games if the weights on the edges as well as the distance between any two strategies are bounded (above and below) by polynomials in the size of the input $\mcI$, then the equilibrium can be computed in polynomial time by best-response dynamics. In this case, the potential is bounded from above by a polynomial in $|\mcI|$, and in each best-response step, the potential also reduces by a polynomial in $|\mcI|$. Hence best-response dynamics converges to a local minima of the potential function in polynomial time, which is also an equilibrium.

Despite these results, we show that equilibrium computation is in general hard in discrete preference games, even in simple settings. Specifically, we show that in the unweighted setting, for any $\alpha > 1/2$, computing an equilibrium is PLS-complete even when each player has constant degree. In the weighted setting, computing an equilibrium is PLS-complete even when each player has constant degree, the number of strategies is constant, and the distance between every pair of strategies is either $1$ or $2$. For directed neighbourhood graphs, we show that an equilibrium may not even exist.

For the hardness results, we show a reduction from the local max-cut game. In a local max-cut game, we are given an undirected weighted graph with $n$ vertices. Vertices correspond to players, and each player has two strategies $A$ and $B$. The utility of a player $i$ is the sum of weights of edges to players that choose the strategy different from $i$, i.e., $u_i(z) = \sum_{j \in N_i: z_j \neq z_i} w_{ij}$. Equilibrium computation in the max-cut game is known to be PLS-complete, even if each player has degree five~\cite{max-cut_pls}. 

\begin{figure}[ht]
\centering 
\includegraphics[width=0.97\textwidth]{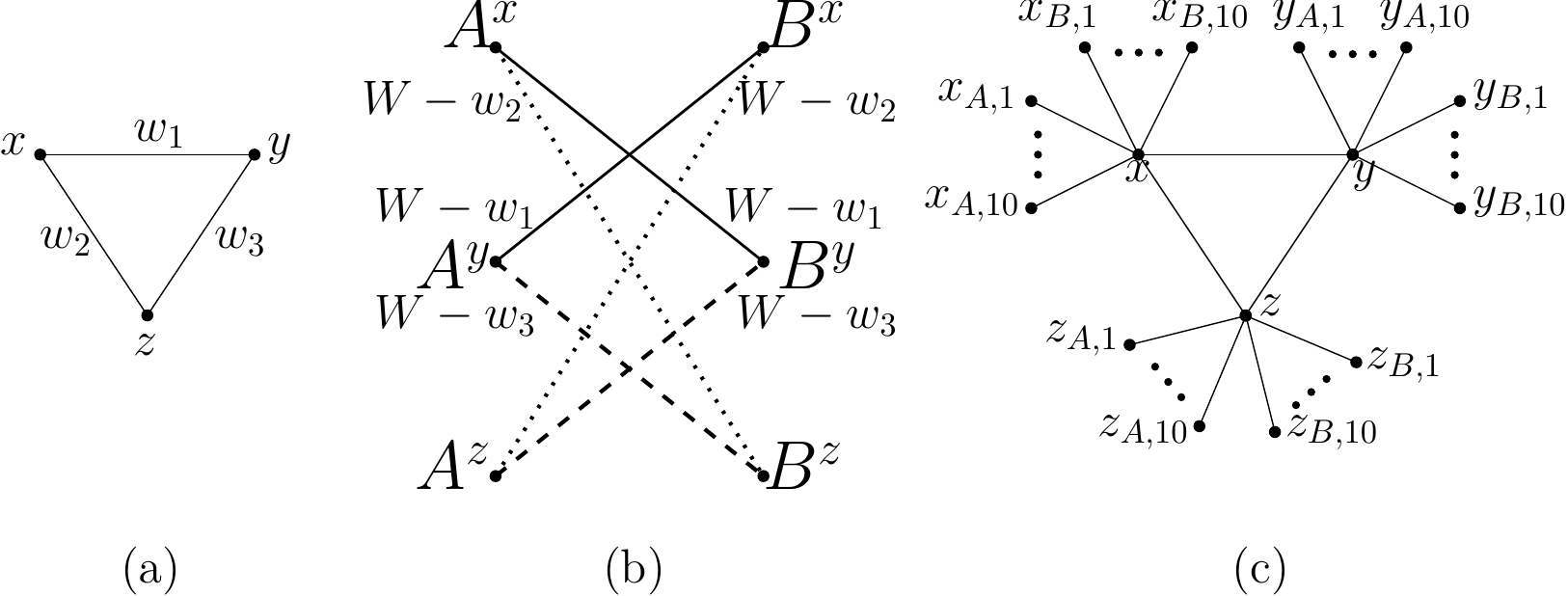}
\caption{An example of the reduction from Theorem~\ref{thm:unweightedhard} (in the theorem, each player has five neighbours in the max-cut instance). Figure (a) shows the max-cut instance, (b) shows the metric space in the reduction, and (c) shows the neighbourhood graph. In (b), all pairs of nodes that do not have an edge displayed between them are at distance $W$, where $W = 5(w_1+w_2+w_3)$.}
\label{fig:example_reduction}
\end{figure}

\begin{theorem}
    For any $\alpha > 1/2$ in the unweighted setting, it is PLS-hard to find an equilibrium in discrete preference games, even when each player has constant degree.
    \label{thm:unweightedhard}
\end{theorem}

\begin{proof}
    Given an instance $G' = (V',E')$ of local max-cut with weights $w'$ on the edges, we construct an instance of a discrete preference game where the strategies are in correspondence with the local max-cut game, and in fact the cost in the discrete preference game is exactly a constant minus the utility in the max-cut game. Let $n'$ be the number of players in either game, and $W = 5 \sum_{e \in E'} w_e'$. We make two assumptions: that each player can be restricted to a subset of strategies, and that some players do not have a preferred strategy. We first describe the reduction under these assumptions, and later show how these assumptions can be removed. With these assumptions, we choose our neighbourhood graph $G(V,E) = G'(V',E')$. The strategy set $L$ contains two strategies $A^i$ and $B^i$ for each player $i$. We assume that $i$ is restricted to these two strategies. Thus, $|L| = 2n'$.  Finally, if $\{i,j\} \in E$, then $d(A^i, B^j)$ = $d(A^j, B^i)$ $= W - w_{ij}$. The distance between any other pair of strategies is $W$. Thus if players $i$ and $j$ both play $A^i$ and $A^j$, or $B^i$ and $B^j$, their distance is $W$.
    
    Figure~\ref{fig:example_reduction} shows the reduction for an instance of max-cut with three vertices $x$, $y$ and $z$.
    
    Note first that the set of players is identical in both games. For every strategy profile $z'$ in the max-cut game, there is a strategy profile $z$ in the discrete preference game where player $i$ plays $A^i$ if she plays $A$ in the max-cut game, and plays $B^i$ otherwise. Then it is easy to see that the cost of player $i$ is $c_i(z) = 5 W - u_i(z')$. There is thus a correspondence between strategy profiles in the two games, and the cost in one is a constant minus the utility in the other. It follows that $z'$ is an equilibrium in the max-cut game if and only if $z$ (as constructed above) is an equilibrium in the discrete preference game.

    We now discuss how to remove the two assumptions. Our first assumption is that a player can be restricted to two strategies. To remove this, for each player $i$, we introduce 20 players: $i_{A,1}, \ldots, i_{A,10}$, and $i_{B,1}, \ldots, i_{B,10}$. We call these \emph{auxiliary} players. Each of these has an edge to player $i$ in the neighbourhood graph, and thus has degree 1. Auxiliary players $i_{A,1}, \ldots, i_{A,10}$ have $A^i$ as their preferred strategy, while auxiliary players $i_{B,1}, \ldots, i_{B,10}$ have $B^i$ as their preferred strategy. Since they have degree 1, and $\alpha > 1/2$, the best response for these players is always to play their preferred strategy. Now note that since each non-auxiliary player $i$ has degree 25 in the neighbourhood graph, if player $i$ plays either $A^i$ or $B^i$, her cost is at most $15W$. However if player $i$ plays a strategy other than $A^i$ or $B^i$, her cost is at least $20(W - \max_e w_e)$ $\ge  16 W$. Hence her best response is always to play either $A^i$ or $B^i$. Further, since the auxiliary players for player $i$ are equally distributed with $A^i$ or $B^i$ as the preferred strategy, their addition does not affect player $i$s choice of strategy between the two, which depends on the strategies chosen by the non-auxiliary players.
    
    Our last assumption is that the non-auxiliary players do not have a preferred strategy. This is removed by introducing another point $C$ into the metric space, which has distance $W$ from all other strategies, and which is the preferred strategy for all non-auxiliary players. However, if $\alpha$ is very large, then it would be an equilibrium for all players to choose $C$. To fix this, increase the number of auxiliary players for each player $i$ from $20$ to $\lceil 20\alpha/(1-\alpha) \rceil$. It can be checked that in this case, player $i$s best response is always to play either $A^i$ or $B^i$. We note that, each player now has degree at most $\lceil 5+20\alpha/(1-\alpha) \rceil$, which is a constant for fixed $\alpha$.
\end{proof}

We now show that if the edges in the neighbourhood graph are weighted, equilibrium computation is hard even in simpler settings.

\begin{theorem}
    In the weighted setting, it is PLS-hard to compute an equilibrium, even when each player has constant degree in the neighbourhood graph, the strategy set has constant size, and the distance between any pair of strategies is either one or two.
    \label{thm:weightedhard}
\end{theorem}

\begin{proof}
    As before, given an instance $G' = (V',E')$ of local max-cut with weights $w'$ on the edges and degree five for each vertex, we construct an instance of a discrete preference game where the strategy profiles are in correspondence with the local max-cut instance. Let $n' = |V'|$ be the number of players and $W = 5 \sum_{e \in E'} w_e'$. We first describe the reduction under the assumption that each player can choose one of only two strategies, and later show how the assumption can be enforced without loss of generality. With this assumption, we choose our weighted neighbourhood graph $G(V,E,w) = G'(V',E',w')$. 
    
    To construct the metric space, we use the fact that a graph of maximum degree five can be properly coloured by a greedy algorithm with six colours. That is, every vertex in the graph can be assigned one of six colours, so that if vertices $u$, $v$ are adjacent in the graph, then they are assigned different colours. Thus, the neighbourhood graph can be coloured with six colours. Let $\kappa(v)$ denote the colour assigned to vertex $v \in V$. Let $a$, $b$, $c$, $d$, $e$, and $f$ be the six colours used.
    
    Our metric space $\mcL$ consists of 12 strategies, $\{A,B\} \times \{a,b,c,d,e,f\}$. We call the first component the \emph{parity} of the strategy, and the second component the \emph{colour} of the strategy. The distance between two points is 1 if the parity of the points is different, and is 2 otherwise. We assume that each player $i$ is restricted to the two strategies in the metric space coloured $\kappa(i)$. Note that this means that for a player $i$, since all of her neighbours have a different colour, they cannot be at the same point in the metric space as $i$. Hence the cost of $i$ is at least $\sum_{j \in N_i} w_{ij}$. Further, it is easily seen that in any strategy profile, the cost of a player $i$ is $2 \sum_{j \in N_i} w_{ij}$ minus the weight of the neighbours of $i$ that play the parity different from $i$'s strategy.
    
    For every strategy profile $z'$ in the max-cut game, there is a strategy profile $z$ in the discrete preference game where player $i$ plays $(A,\kappa(i))$ if she plays $A$ in the max-cut game, and plays $(B,\kappa(i))$ otherwise. Then the cost of player $i$ in the discrete preference game is $c_i(z) = 2 \sum_{j \in N_i} w_{ij} - u_i(z')$. There is thus a correspondence between strategy profiles in the two games, and the cost in one is (a constant plus) the negative of the utility in the other. It follows that $z'$ is an equilibrium in the max-cut game if and only if $z$ (as constructed above) is an equilibrium in the discrete preference game.
    
    We remove the assumption in a manner similar to the previous proof, though since the neighbourhood graph is weighted we require fewer auxiliary players. For each existing player $i$, we introduce 2 new players $i_A$ and $i_B$, called auxiliary players. These players have an edge from player $i$ in the neighbourhood graph with weight $W := 10 \sum_{e \in E} w_e$. Each auxiliary player thus has degree 1. Then auxiliary player $i_A$ has $(A,\kappa(i))$ as its preferred strategy with weight $11W$, and $i_B$ has $(B,\kappa(i))$ as its preferred strategy with weight $11W$. Notice that: (1) Since the auxiliary players have degree 1 with an edge of weight $10W$ incident, while the weight they place on their preferred strategy is $11W$, they will always play their preferred strategy. (2) By a simple calculation as in the previous proof, the best response for player $i$ is always to play either $(A,\kappa(i))$ or $(B,\kappa(i))$. The symmetry of the auxiliary players implies that their presence does not affect the choice of $(A,\kappa(i))$ or $(B,\kappa(i))$ for player $i$. This completes the proof.
\end{proof}

We now give an example for a directed neighbourhood graph where an equilibrium does not exist. As before, we first describe our example under the assumption that we can restrict players to a subset of strategies, and then introduce auxiliary players to remove this assumption.	
	
			\begin{figure}[ht]
				\includegraphics[width=.97 \textwidth]{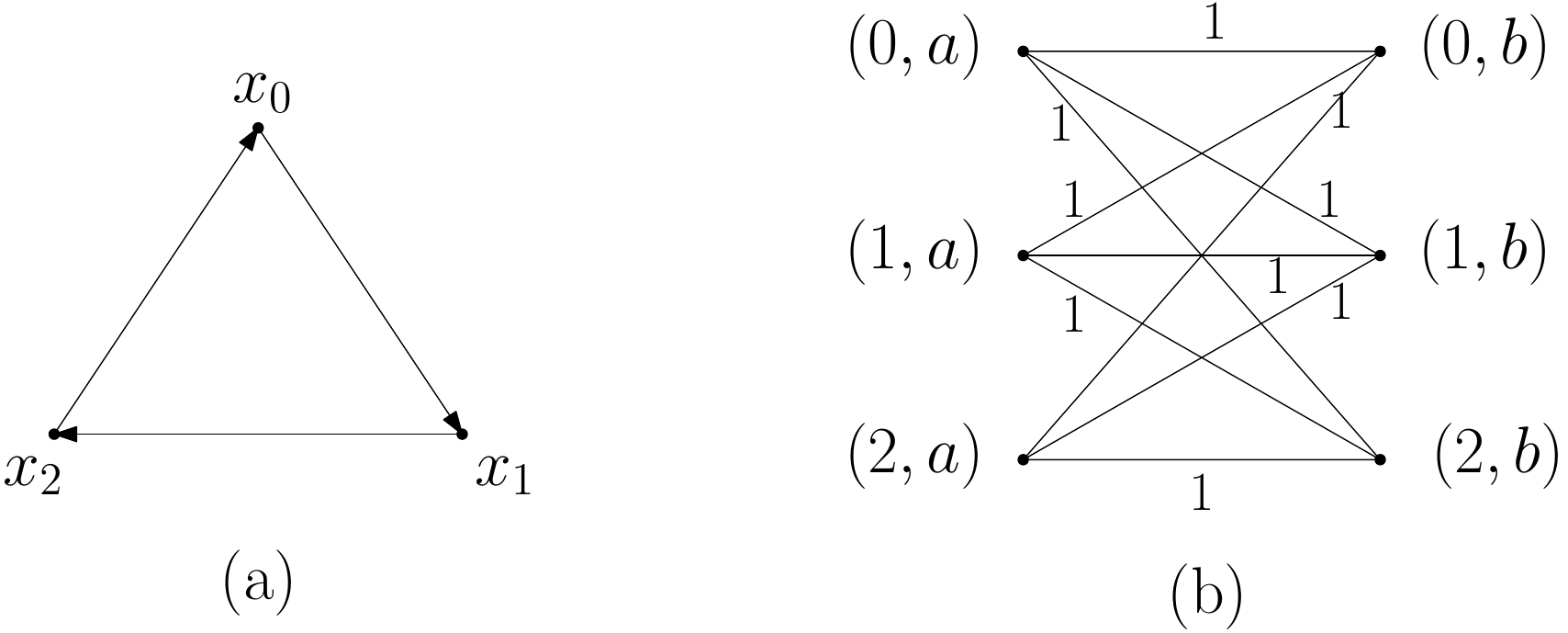}
				\centering
				\caption{The example for non-existence of equilibria in directed discrete preference games. Figure(a) shows the neighborhood graph and  (b) shows the metric space.}
				\label{fig:counter_example}
			\end{figure}

\begin{example}
	With the assumption that we can restrict players to a subset of the strategies, the neighbourhood graph and the metric space for our example are shown in Figure~\ref{fig:counter_example}. There are three players $x_0$, $x_1$, and $x_2$. In the neighbourhood graph, player $x_i$ has an edge to $x_{i+1 \bmod 3}$ (in this example, we always assume $i+1$ is taken $\bmod$ $3$ to avoid repetition). The metric space consists of 6 strategies, $\{0,1,2\} \times \{a,b\}$. We think of the second coordinate as the `parity', and strategies as nodes in a complete bipartite graph with $6$ vertices, with any two strategies of the same parity at distance 2, while any two strategies of different parities at distance 1. By our assumption, we restrict player $x_i$ to strategies $(i,a)$ and $(i,b)$.
	
	From the neighbourhood graph, player $x_i$ wants to be near player $x_{i+1}$ for $i \in \{0,1,2\}$. However, in the metric space, for any (restricted) choice of strategy for $x_{i+1}$, the strategy of $x_i$ that is nearest has the opposite parity. Hence each player $x_i$ tries to choose a strategy of the opposite parity from player $x_{i+1}$, and hence there is no equilibrium.
	
	Lastly, to remove the assumption of strategy restrictions, for each player $x_i$, we add ten new players $(x_i,1), \ldots, (x_i,10)$. In the neighbourhood graph each player $x_i$ has an edge to the ten new players $(x_i,j)$, and the game now has 33 players. For each $i \in \{0,1,2\}$, the first five new players $(x_i,1) \ldots (x_i,5)$ have $(i,a)$ as their preferred strategy, while the other five have $(i,b)$ as their preferred strategy. Since the newly added players only have incoming edges, they always choose their preferred strategy. Then for each player $x_i$, choosing a strategy from the set $\{(i,a), (i,b)\}$ gets cost $5$ from the auxiliary players, and cost at most $4$ from the other two non-auxiliary players. Whereas, any different strategy for player $i$ gets cost at least $15$ from the auxiliary players. Hence player $x_i$ will always choose from the set $\{(i,a), (i,b)\}$ at equilibrium. Since the newly added players are symmetric between these two strategies, they do not further affect $x_i$'s choice of strategy.
\end{example}


\section{Algorithms for Computing Equilibria}
\label{sec:algorithms}


We now give efficient algorithms for computing equilibria in discrete preference games with restrictions on the metric space. However, we allow a significant generalization of the neighbourhood graph. We allow directed, weighted neighbourhood graphs, where instead of a preferred strategy, players have a penalty associated with each point in the metric space. Formally, for each node $v$ in the metric space and each player $i$, there is a real-valued penalty $p_i(v)$. The cost for player $i$ for the strategy profile $z=(z_i,z_{-i})$ is 

\[
c_i(z) = \sum_{v \in L} p_i(v) d(v,z_i) + \sum_{j \in N_i} w_{ij} d(z_i,z_j) \, .
\]

Our results thus show that in the metric spaces discussed below, equilibria exist, even in the case of directed neighbourhood graphs. E.g., this shows that equilibria exist in the case of path metrics.

We discuss metric spaces in more detail now. Any undirected weighted graph on $r$ vertices corresponds to a metric space with $r$ points, where every vertex is a point, and the distance between any pair of points is the weight of the minimum weight path in the graph between the corresponding vertices. Such a metric space is a graph metric. Further, any finite metric space on $r$ points can be represented as a graph metric, by considering the complete graph on $r$ vertices where the weight of the edge between any pair of vertices is the distance between them. 

We first give an algorithm for when the graph metric is a tree, with positive lengths $l_e$ on the edges. Note that this contains the special case when the graph metric is a path. We then generalise path metrics in another direction, by considering the Cartesian product of path metrics. This product metric intuitively is obtained when the graph for the metric space is a regular grid.


\subsection{An algorithm for tree metrics}


\begin{algorithm}[!ht]
\caption{Tree Metric Algo} \label{algo:brd_tree}
\begin{algorithmic}[1]
\Require {Discrete preference game $(G=(V,E,w),\mcT=(L,d))$ where $\mcT$ is a tree metric with root $r$.}
\State {Initially, let $z_i \leftarrow r$ for each player $i$.}
\State {\textbf{while} $\exists$ player $i$ that can reduce her cost by moving to a child $v$ of $z_i$ \textbf{ do } $z_i \leftarrow v$.}
\end{algorithmic}
\end{algorithm}

Our algorithm for tree metrics initially places all players at the root. If any player can improve her cost by moving to a child of her current strategy, the algorithm changes her strategy accordingly. In a metric space with $n$ points, the algorithm terminates in $pn$ iterations of the while loop where $p$ is the number of points, and hence terminates in polynomial time. We now show that when it terminates, the strategy profile is an equilibrium.

To prove convergence, we first characterise the best response. Fix a player $i$ and strategies $z_{-i}$ for the other players. For any node $v$ in the tree, let $w(v)$ be the weight of $i$'s neighbours $j \in N_i$ that have $z_j = v$, plus $i$'s penalty $p_i(v)$ for point $v$. This gives us a tree $T$ with weights on the nodes. We say that cost of node $u$ in the tree is the total weighted distance to the other nodes, i.e., $c(u) = \sum_{v \in T} w(v) \, d(u,v)$. The set of minimum cost nodes in the weighted tree are called the medians of the tree, and are the best responses for player $i$, since $c(u) = c_i(z)$ if $z_i = u$.

We will use the following result which further characterises the medians. Given weights $w(v)$ at the nodes, let $w(T)$ be the total weight at the nodes of the tree, and $T-v$ be the graph obtained by removing node $v$.

\begin{claim}[\cite{chierichetti:jcss18}]
A node $u$ is the median of a tree iff the weight of each connected component of $T-v$ is at most $w(T)/2$.
\end{claim}

\noindent We also use the following claim.

\begin{claim}
Given a tree $T$ with weights at the nodes, let $v$ be an arbitrary vertex and $v^*$ be a median nearest to $v$. Then the cost of the nodes strictly decreases on the path from $v$ to $v^*$.
\label{claim:treemonotone}
\end{claim}

\begin{proof}
    Let $(v^*=v_0, v_1, \ldots, v_t = v)$ be the path from $v^*$ to $v$. Note that $c(v_1) > c(v_0)$, since $v_i$ is not a median for $i > 0$. Also, all of these nodes (except $v_0$, which is the median) are in the same connected component in $T - v^*$, and the total weight of nodes in this component is $\le w(T)/2$. Now consider any node $v_i$ for $i > 1$. We know that the subtree rooted at $v_i$ has total weight at most $w(T)/2$ (since it is in the same component in $T - v^*$). Since $v_{i-1}$ is not a median, the subtree rooted at $v_i$ must have total weight strictly less than $w(T)/2$. Moving from $v_i$ to $v_{i-1}$ increases the distance from every node in this subtree by the length of the edge $(v_{i-1},v_i)$, and decreases the distance from every other node by this quantity, and hence decreases the cost.
\end{proof}

We now prove convergence of the algorithm.

\begin{theorem}
The Tree Metric Algo terminates at an equilibrium.

\begin{proof}
    Let $z$ be the strategy profile when the algorithm terminates. Suppose for a contradiction that for player $i$, $z_i = v$, while $z_i' = v^*$ is a nearest best response(and so a median) with lower cost. In the following, we consider the weighted tree $T$ with edge lengths as in the metric space, and weights on the nodes, where for any node $v$ in the tree the weight $w(v) = \sum_{j \in N_i: z_j = v} w_{ij}$ $+ p_i(v)$. As earlier, the cost of a node $v$ in the tree is the total weighted distance to the other nodes
    
    Let $(v^*=v_0, v_1, \ldots, v_t = v)$ be the path from $v^*$ to $v$, then by Claim~\ref{claim:treemonotone}, the cost strictly increases along this path, and $c(v_{t-1}) < c(v)$. Since the algorithm terminates, $v_{t-1}$ must be $v$'s parent, hence $v \neq r$. Let $T(v)$ be the subtree rooted at $v$. Consider the timestep when player $i$ moved from $v_{t-1}$ to $v$. Note that this decreases $i$'s distance from every node in $T(v)$ by $l_{v_{t-1},v}$, and increases the distance from every other node by the same length. Since this move decreased $i$'s cost, at that time, the total weight of $i$'s neighbours in $T(v)$ must have been at least $w(T)/2$. Since that time step, players have only moved away from the root, and hence in particular any player that was in $T(v)$ at that timestep must still be in $T(v)$, and hence when the algorithm terminates, the weight of $i$'s neighbours in $T(v)$ must be at least $w(T)/2$. However, since $v^*$ is a median, the weight of $i$'s neighbours in $T(v)$ is also at most $w(T)/2$. Thus $v$ must also be a median, giving us a contradiction.
\end{proof}
\end{theorem}


\subsection{An algorithm for the Cartesian product of path metrics}


We now give an algorithm for equilibrium computation if the metric space is the \emph{Cartesian product of path metrics}. As discussed, a \emph{path metric} $\mcP = (L,d)$ can be represented as a path. Alternatively, a path metric can be embedded in the real number line so that the distance between two points is the absolute difference in their values of their embedding. 

A metric space $\mcP = (L,d)$ is the \emph{Cartesian product of path metrics} $\mcP_1 = (L_1, d_1)$, $\ldots$, $\mcP_r = (L_r, d_r)$ (or a product metric, for brevity) if $L = L_1 \times \ldots \times L_r$, and for any two points $x=(x_1, \ldots, x_r)$ and $y = (y_1, \ldots, y_r)$ in $L$, the distance $d(x,y)$ $= \sum_{i=1}^r d_i(x_i,y_i)$. Alternatively, $\mcP$ is the Cartesian product of $r$ path metrics if it can be embedded in $\real^r$, so that the distance between any two points is the $L_1$ distance of their embeddings, and whenever $(x_1, \ldots, x_r)$ and $(y_1, \ldots, y_r)$ are points in the embedding, so are the $2^r$ $\{x_1,y_1\}$ $\times \ldots \times$ $\{x_r,y_r\}$. 

For a discrete preference game on a product metric, for each player $i$, her strategy $z_i$ is a vector, with the $t$th coordinate $z_{i,t}$ denoting her position in the path metric $\mcP_t$.

For the algorithm, we first characterize equilibria. Given a discrete preference game with product metric $\mcP$ and a strategy profile $z$, we say player $i$ is playing her \emph{partial best response in the $t$th metric}
if she is at a median in the path metric $\mcP_t$ (we defined the set of medians earlier, for tree metrics). Note that a player may have multiple best responses.

\begin{claim}
Player $i$ is playing her best response iff she is playing her partial best response in each metric $t \in [r]$.
\label{claim:brdcharacter1}
\end{claim}

\begin{proof}
    The claim is because the distance between two points in the product metric space $d(x,y)$ $= \sum_{i=1}^r d_i(x_i,y_i)$ is the sum of distances along the individual metric spaces. Further, the position in each path metric can be chosen independently. Hence a player minimizes her cost if and only if she minimizes her cost in each component path metric, i.e., she plays a partial best response in each path metric $\mcP_t$ for $t \in [r]$.
\end{proof}

\begin{algorithm}[!ht]
\caption{Product Metric Algo} \label{algo:brd_reald}
\begin{algorithmic}[1]
\Require {Discrete preference game $(G=(V,E,w),\mcP=(K,d))$ where $\mcP$ is the Cartesian product of path metrics $\mcP_1$, $\ldots$, $\mcP_r$.}
\State {Initially, let $z_i \leftarrow s_i$ for each player $i$.}
\For{$k \in [r]$}
    \State {Use algorithm Tree Metric Algo to obtain an equilibrium for the players in the path metric $\mcP_k$.}
    \State {For each player, set $z_{i,k}$ to her position in $\mcP_k$ in the equilibrium computed.}
\EndFor
\end{algorithmic}
\end{algorithm}

Since the Tree Metric Algo terminates in polynomial time, so does the Product Metric Algo.

\begin{theorem}
The Product Metric Algo terminates at an equilibrium.
\end{theorem}

\begin{proof}
    Let $z$ be the strategy profile when the algorithm terminates. For each player $i$, the $t$th component $z_{i,t}$ is set in the $t$th iteration of the for loop, and after this iteration each player is playing a partial best response in $\mcP_t$. By Claim~\ref{claim:brdcharacter1}, after the last iteration of the for loop, the strategy profile is an equilibrium.
\end{proof}

\subsection*{Conclusion} 

Our work is the first to study the basic question of efficient equilibrium computation in discrete preference games. We show that despite incentivizing coordination, in general equilibrium computation is PLS-hard. However with restrictions on the metric space, equilibrium may be computed in polynomial time, even in very general settings for the neighbourhood graph. Our work is a first step, and leaves open many interesting problems. As an example, for what other metric spaces can we find an equilibrium efficiently? Another interesting direction would be to place restrictions on the neighbourhood graph to better represent real-life social networks, and study if these make equilibrium computation any easier. With the growing popularity of this and other models of opinion formation, we feel these are important, fundamental questions.

\subsection*{Acknowledgement}
We thank Harit Vishwakarma and Rakesh Pimplikar for interesting discussions at the initial stages of this project.

\appendix

\bibliography{ijcai19}

\begin{thebibliography}{10}

\bibitem{AcemogluDLO11}
Daron Acemoglu, Munther~A Dahleh, Ilan Lobel, and Asuman Ozdaglar.
\newblock Bayesian learning in social networks.
\newblock {\em The Review of Economic Studies}, 78(4):1201--1236, 2011.

\bibitem{AptKRSS17}
Krzysztof~R. Apt, Bart de~Keijzer, Mona Rahn, Guido Sch{\"{a}}fer, and Sunil
  Simon.
\newblock Coordination games on graphs.
\newblock {\em Int. J. Game Theory}, 46(3):851--877, 2017.

\bibitem{AptSW16}
Krzysztof~R. Apt, Sunil Simon, and Dominik Wojtczak.
\newblock Coordination games on directed graphs.
\newblock In {\em Proceedings Fifteenth Conference on Theoretical Aspects of
  Rationality and Knowledge, {TARK} 2015, Carnegie Mellon University,
  Pittsburgh, USA, June 4-6, 2015.}, pages 67--80, 2015.

\bibitem{AulettaCFGP16}
Vincenzo Auletta, Ioannis Caragiannis, Diodato Ferraioli, Clemente Galdi, and
  Giuseppe Persiano.
\newblock Generalized discrete preference games.
\newblock In {\em Proceedings of the Twenty-Fifth International Joint
  Conference on Artificial Intelligence, {IJCAI} 2016, New York, NY, USA, 9-15
  July 2016}, pages 53--59, 2016.

\bibitem{BalaG98}
Venkatesh Bala and Sanjeev Goyal.
\newblock Learning from neighbours.
\newblock {\em The Review of Economic Studies}, 65(3):595--621, 1998.

\bibitem{BarabasiA99}
Albert-L{\'a}szl{\'o} Barab{\'a}si and R{\'e}ka Albert.
\newblock Emergence of scaling in random networks.
\newblock {\em Science}, 286(5439):509--512, 1999.

\bibitem{bindel-kleinberg:15}
David Bindel, Jon~M. Kleinberg, and Sigal Oren.
\newblock How bad is forming your own opinion?
\newblock {\em Games and Economic Behavior}, 92:248--265, 2015.

\bibitem{BoykovVZ01}
Yuri Boykov, Olga Veksler, and Ramin Zabih.
\newblock Fast approximate energy minimization via graph cuts.
\newblock {\em {IEEE} Trans. Pattern Anal. Mach. Intell.}, 23(11):1222--1239,
  2001.
\newblock URL: \url{https://doi.org/10.1109/34.969114}, \href
  {http://dx.doi.org/10.1109/34.969114} {\path{doi:10.1109/34.969114}}.

\bibitem{CaiD11}
Yang Cai and Constantinos Daskalakis.
\newblock On minmax theorems for multiplayer games.
\newblock In {\em Proceedings of the Twenty-Second Annual {ACM-SIAM} Symposium
  on Discrete Algorithms, {SODA} 2011, San Francisco, California, USA, January
  23-25, 2011}, pages 217--234, 2011.

\bibitem{dpg:13}
F.~Chierichetti, J.~Kleinberg, and S.~Oren.
\newblock On discrete preferences and coordination.
\newblock In {\em Proceedings of the 14th ACM Conference on Electronic Commerce
  (ACM EC)}, pages 233--250, 2013.

\bibitem{chierichetti:jcss18}
Flavio Chierichetti, Jon~M. Kleinberg, and Sigal Oren.
\newblock On discrete preferences and coordination.
\newblock {\em J. Comput. Syst. Sci.}, 93:11--29, 2018.
\newblock URL: \url{https://doi.org/10.1016/j.jcss.2017.11.002}, \href
  {http://dx.doi.org/10.1016/j.jcss.2017.11.002}
  {\path{doi:10.1016/j.jcss.2017.11.002}}.

\bibitem{CliffordS73}
Peter Clifford and Aidan Sudbury.
\newblock A model for spatial conflict.
\newblock {\em Biometrika}, 60(3):581--588, 1973.

\bibitem{DaskalakisGP09}
Constantinos Daskalakis, Paul~W. Goldberg, and Christos~H. Papadimitriou.
\newblock The complexity of computing a nash equilibrium.
\newblock {\em Commun. {ACM}}, 52(2):89--97, 2009.
\newblock URL: \url{https://doi.org/10.1145/1461928.1461951}, \href
  {http://dx.doi.org/10.1145/1461928.1461951}
  {\path{doi:10.1145/1461928.1461951}}.

\bibitem{DomingosR01}
Pedro Domingos and Matt Richardson.
\newblock Mining the network value of customers.
\newblock In {\em Proceedings of the Seventh ACM SIGKDD International
  Conference on Knowledge Discovery and Data Mining (KDD '01)}, pages 57--66.
  ACM, 2001.

\bibitem{max-cut_pls}
Robert Els{\"a}sser and Tobias Tscheuschner.
\newblock Settling the complexity of local max-cut (almost) completely.
\newblock In {\em ICALP}, 2011.

\bibitem{FabrikantPT04}
Alex Fabrikant, Christos~H. Papadimitriou, and Kunal Talwar.
\newblock The complexity of pure nash equilibria.
\newblock In {\em Proceedings of the 36th Annual {ACM} Symposium on Theory of
  Computing, Chicago, IL, USA, June 13-16, 2004}, pages 604--612, 2004.

\bibitem{FerraioliGV16}
Diodato Ferraioli, Paul~W. Goldberg, and Carmine Ventre.
\newblock Decentralized dynamics for finite opinion games.
\newblock {\em Theor. Comput. Sci.}, 648:96--115, 2016.
\newblock URL: \url{https://doi.org/10.1016/j.tcs.2016.08.011}, \href
  {http://dx.doi.org/10.1016/j.tcs.2016.08.011}
  {\path{doi:10.1016/j.tcs.2016.08.011}}.

\bibitem{golub:2010}
B.~Golub and M.O. Jackson.
\newblock Naïve learning in social networks: Convergence, influence, and the
  wisdom of crowds.
\newblock {\em American Economics Journal: Microeconomics}, 2(1):112--149,
  2010.

\bibitem{JohnsonPY88}
David~S. Johnson, Christos~H. Papadimitriou, and Mihalis Yannakakis.
\newblock How easy is local search?
\newblock {\em J. Comput. Syst. Sci.}, 37(1):79--100, 1988.

\bibitem{KempeKT15}
David Kempe, Jon~M. Kleinberg, and {\'{E}}va Tardos.
\newblock Maximizing the spread of influence through a social network.
\newblock {\em Theory of Computing}, 11:105--147, 2015.
\newblock URL: \url{https://doi.org/10.4086/toc.2015.v011a004}, \href
  {http://dx.doi.org/10.4086/toc.2015.v011a004}
  {\path{doi:10.4086/toc.2015.v011a004}}.

\bibitem{krackhardt:09}
D.~Krackhardt.
\newblock A plunge into networks.
\newblock {\em Science}, 326:47--48, 2009.

\bibitem{SchafferY91}
Alejandro~A. Sch{\"{a}}ffer and Mihalis Yannakakis.
\newblock Simple local search problems that are hard to solve.
\newblock {\em {SIAM} J. Comput.}, 20(1):56--87, 1991.
\newblock URL: \url{https://doi.org/10.1137/0220004}, \href
  {http://dx.doi.org/10.1137/0220004} {\path{doi:10.1137/0220004}}.

\bibitem{WattsS98}
Duncan~J Watts and Steven~H Strogatz.
\newblock Collective dynamics of ‘small-world’ networks.
\newblock {\em Nature}, 393(6684):440, 1998.

\bibitem{yildiz:2013}
Mehmet~Ercan Yildiz, Asuman~E. Ozdaglar, Daron Acemoglu, Amin Saberi, and Anna
  Scaglione.
\newblock Binary opinion dynamics with stubborn agents.
\newblock {\em ACM Trans. Economics and Comput}, 1(4), 2013.

\end{thebibliography}

\end{document}